\documentclass[11pt]{llncs}
\pdfoutput=1
\usepackage{graphicx}

\usepackage{times}
\usepackage{mathptm}
\usepackage{cite}

\DeclareSymbolFont{AMSb}{U}{msb}{m}{n}
\DeclareSymbolFontAlphabet{\Bbb}{AMSb}

\makeatletter
\def\hb@xt@{\hbox to }
\makeatother

\let\oldendproof\endproof
\def\endproof{\qed\oldendproof}

\begin{document}
\title{Finding Large Clique Minors is Hard} 

\author{David Eppstein}

\institute{Computer Science Department\\
University of California, Irvine\\
\email{eppstein@uci.edu}}

\maketitle   

\begin{abstract}
We prove that it is NP-complete, given a graph $G$ and a parameter $h$, to determine whether $G$ contains a complete graph $K_h$ as a minor.
\end{abstract}

\section{Introduction}

The \emph{Hadwiger number} of a graph $G$ is the number of vertices in the largest clique that is a \emph{minor} of $G$; that is, that can be formed by contracting some edges and deleting others. Equivalently, it is the largest number of vertex-disjoint connected subgraphs that one can find in $G$ such that for each two subgraphs $S_i$ and $S_j$ there is an edge $v_iv_j$ in $G$ with $v_i\in S_i$ and $v_j\in S_j$. In 1943, Hugo Hadwiger conjectured that in any graph the Hadwiger number is greater than or equal to the chromatic number~\cite{Had-VNZ-43}, and this important conjecture remains open in general, although it is known to be true when the chromatic number is at most six~\cite{RobSeyTho-C-93}. The Hadwiger number is also closely associated with the sparseness of the given graph: if $G$ has Hadwiger number $h$, every subgraph of $G$ has a vertex with degree $O(h\sqrt{\log h})$. It follows from this fact that, if $G$ has $n$ vertices, it has $O(nh\sqrt{\log h})$ edges~\cite{Kos-C-94}.

Given its graph-theoretic importance, it is natural to ask for the computational complexity of the Hadwiger number. In this light, Alon et al.~\cite{AloLinWah-TCS-07} observe that the Hadwiger number is fixed parameter tractable: for any constant $h$, there is a polynomial-time algorithm that either computes the Hadwiger number or determines that it is greater than $h$, and the exponent in the polynomial time bound of this algorithm is independent of $h$, due to standard results in graph minor theory. However, this is not a polynomial time algorithm for the Hadwiger number problem because its running time includes a factor exponential or worse in $h$. In addition, as Alon et al. show, the Hadwiger number may be approximated in polynomial time more accurately than the problem of finding the largest clique subgraph of a given graph: they provide a polynomial time approximation algorithm for the Hadwiger number of an $n$-vertex graph with approximation ratio $O(\sqrt n)$, whereas it is NP-hard to approximate the clique number to within a factor better than $n^{1-\epsilon}$ for any $\epsilon>0$~\cite{Zuc-STOC-06}.

To classify the problem of computing the Hadwiger number in complexity theoretic terms, we need to consider a decision version of the problem: given a graph $G$, and a number $h$, is the Hadwiger number of $G$ greater than or equal to $h$?\footnote{Chandran and Sivadasan~\cite{ChaSiv-DM-07} formulate a different decision problem, in which the positive instances are those with small Hadwiger number, but this is in CoNP rather than being in NP.} We call this decision problem the \emph{Hadwiger number problem}. Unsurprisingly, it turns out to be NP-complete. A statement of its NP-hardness was made without proof by Chandran and Sivadasan~\cite{ChaSiv-DM-07}. However, we have been unable to find a clear proof of the NP-completeness of the Hadwiger number problem in the literature. Our goal in this short paper is to fill this gap by providing an NP-completeness proof of the standard type (a polynomial-time many-one reduction from a known NP-complete problem) for the Hadwiger number problem.

\section{Reduction from domatic number}

\begin{figure}[t]
\centering\includegraphics[width=5.5in]{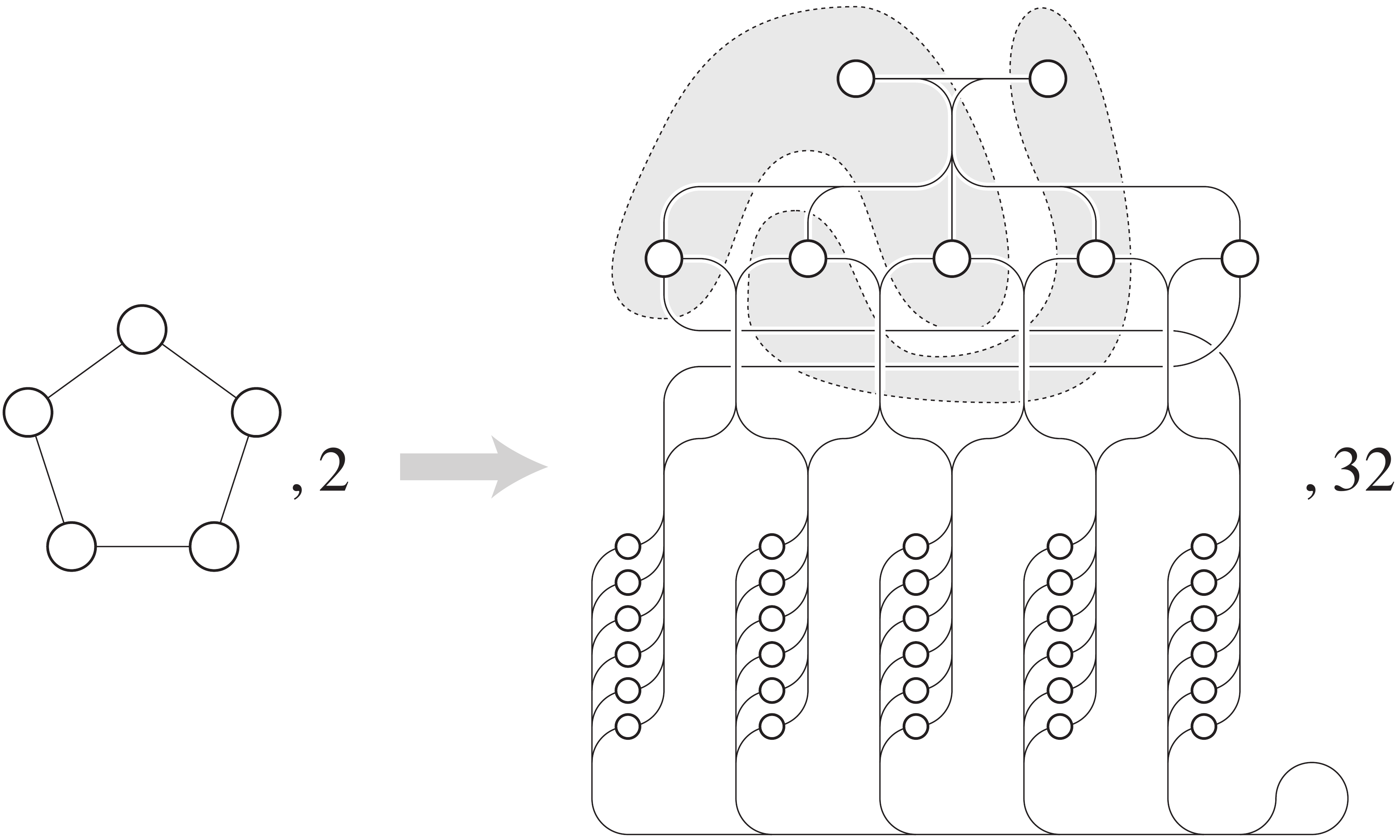}
\caption{A confluent drawing~\cite{DicEppGoo-GD-03} of our NP-completeness reduction. Two vertices are connected by an edge in $G'$ if and only if there is a smooth (possibly self-intersecting) path between the circles representing them in the drawing. In this example, a 5-cycle with domatic number~2 is transformed into a 37-vertex graph with Hadwiger number 32. One possible 32-vertex clique minor is formed by contracting the two shaded sets of vertices into single supervertices, removing the remaining middle-layer vertex, and combining the two supervertices with the 30 bottom-layer vertices.}
\label{fig:redux}
\end{figure}

Recall that a vertex $v$ \emph{dominates} a vertex $w$ if $v=w$ or $v$ and $w$ are adjacent; a dominating set of a graph $G$ is a set of vertices such that, for every vertex $w$ in $G$, some member of the set dominates $w$. The \emph{domatic number} of a graph $G$ is the maximum number of disjoint dominating sets that can be found in $G$~\cite{CocHed-Nw-77}. In the \emph{domatic number problem}, we are given a graph $G$ and a number $d$, and asked to determine whether the domatic number of $G$ is at least $d$; that is, whether $G$ contains at least $d$ disjoint dominating sets. This problem is known to be NP-complete even for $d=3$: that is, it is difficult to determine whether the vertices of a given graph may be partitioned into three dominating sets~\cite{GarJoh-79}.

We begin by describing a polynomial-time many-one reduction from an instance $(G,d)$ of the domatic number problem into an instance $(G',h)$ of the Hadwiger number problem. We may assume without loss of generality that no vertex of $G$ is adjacent to all others, for if $v$ is such a vertex we may take one of the dominating sets to be the one-vertex set $\{v\}$, simplify the problem in polynomial time by deleting $v$ from $G$ and subtracting one from $d$, and use the remaining simplified problem as the basis for the transformation. As we will show, with this assumption, the instance $(G,d)$ may be translated in polynomial time to an equivalent instance $(G',h)$ of the Hadwiger number problem.
We also assume that the vertices of $G$ are numbered arbitrarily as $v_i$ for $1\le i\le n$.
To perform this translation, construct $G'$ in three layers:
\begin{itemize}
\item The top layer is a $d$-vertex clique with vertices $t_i$ for $1\le i\le d$.
\item The middle layer is an $n$-vertex independent set with vertices $m_i$ for $1\le i\le n$.
\item The bottom layer is an $n(n+1)$-vertex clique with vertices $b_{i,j}$ for $1\le i\le n$ and $1\le j\le n+1$.
\item Every pair of one top and one middle vertex is connected by an edge.
\item Middle vertex $m_i$ and bottom vertex $b_{j,k}$ are connected by an edge if and only if either $i=j$ or $G$ has an edge $v_iv_j$. That is, there is an edge from $m_i$ to $b_{j,k}$ if and only if $v_i$ dominates $v_j$.
\end{itemize}
We let $h = n(n+1)+d$. This reduction is illustrated in Figure~\ref{fig:redux}.

\begin{lemma}
\label{lem:connectivity}
Let $G'$ be as constructed above, and let $S$ be a connected nonempty subset of the vertices of $G'$. Then at least one of the following three possibilities is true:
\begin{enumerate}
\item $S$ consists only of a single middle vertex.
\item $S$ contains a top vertex.
\item $S$ contains a bottom vertex.
\end{enumerate}
\end{lemma}

\begin{proof}
If $S$ consists only of middle vertices, it can have only one of them, because the middle vertices form an independent set. If on the other hand $S$ does not consist only of middle vertices, it must contain a top vertex or a bottom vertex.
\end{proof}

\begin{lemma}
\label{lem:degree}
Let vertex $v_i$ in $G$ have degree $d_i$. Then middle vertex $m_i$ in $G'$ has degree $(d_i+1)(n+1)+d$.
\end{lemma}

\begin{proof}
Vertex $m_i$ is connected to $(d_i+1)(n+1)$ bottom vertices: the $n+1$ vertices $b_{i,j}$ (for $1\le j\le n+1$) and the $d_i(n+1)$ vertices $b_{i',j}$ (for $1\le j\le n+1$) such that $v_i$ and $v_i'$ are neighbors.
In addition it is connected to all $d$ top vertices.
\end{proof}

\begin{lemma}
\label{lem:topbot}
Let $G'$ be as constructed above, suppose that $G'$ has Hadwiger number at least $h$, and let $S_i$ ($1\le i\le h$) be a family of disjoint mutually-adjacent subgraphs forming an $h$-vertex clique minor in $G'$. Then each subgraph $S_i$ has exactly one non-middle vertex and each non-middle vertex belongs to exactly one subgraph $S_i$.
\end{lemma}

\begin{proof}
A middle vertex in $G'$ has degree at most $(n-1)(n+1)+d<h-1$ by Lemma~\ref{lem:degree} and by the assumption that $G$ has no vertex that is adjacent to all other vertices. If there were a subgraph $S_i$ consisting only of a single middle vertex, it would not have enough neighbors to be adjacent to all $h-1$ of the other subgraphs, so we may infer that such subgraphs do not exist and apply Lemma~\ref{lem:connectivity} to conclude that each subgraph $S_i$ contains at least one non-middle vertex. But there are $h$ subgraphs, and $h$ non-middle vertices, so each subgraph $S_i$ must contain exactly one non-middle vertex and each such vertex must belong to one of these subgraphs.
\end{proof}

\begin{lemma}
\label{lem:isolani}
Let $G'$ be as constructed above, suppose that $G'$ has Hadwiger number at least $h$, and let $S_i$ ($1\le i\le h$) be a family of disjoint mutually-adjacent subgraphs forming an $h$-vertex clique minor in $G'$. Then, for each $i$, there is a bottom vertex $b_{i,j}$ that forms a single-vertex subgraph in the family.
\end{lemma}

\begin{proof}
By Lemma~\ref{lem:topbot}, each set that contains more than one vertex contains a middle vertex.
But there are only $n$ middle vertices, so at most $n$ disjoint sets can contain middle vertices.
Since there are $n+1$ bottom vertices $b_{i,j}$, and (by Lemma~\ref{lem:topbot} again) each belongs to a different subgraph, one of the $n+1$ subgraphs containing these bottom vertices must have no middle vertices. Since it contains only one non-middle vertex, it must form a single-vertex subgraph. 
\end{proof}

\begin{lemma}
\label{lem:dom}
Let $G'$ be as constructed above, suppose that $G'$ has Hadwiger number at least $h$, let $S_i$ ($1\le i\le h$) be a family of disjoint mutually-adjacent subgraphs forming an $h$-vertex clique minor in $G'$, and let $t_i$ be a top vertex belonging to set $S_i$. Then the set $D_i=\{v_j\mid m_j\in S_i\}$ is a dominating set in $G$.
\end{lemma}

\begin{proof}
Let $v_k$ be any vertex in $G$, and let $b_{k,k'}$ be a bottom vertex in $G'$ that forms a single-vertex subgraph in the family of disjoint subgraphs; $b_{k,k'}$ is guaranteed to exist by Lemma~\ref{lem:isolani}. Then $S_i$ must contain a vertex adjacent to $b_{k,k'}$; this vertex must be a middle vertex $m_j$ of $G'$, because top vertices are not adjacent to bottom vertices and $S_i$ contains top vertex $t_i$ as its only non-middle vertex. In order for middle vertex $m_j$ to be adjacent to bottom vertex $b_{k,k'}$, the vertex $v_j$ in $G$ that corresponds to $m_j$ must dominate $v_k$. Thus, for every vertex $v_k$ in $G$, there is a vertex $v_j$ in $D_i$ that dominates $v_k$; therefore, $D_i$ is a dominating set.
\end{proof}

\begin{lemma}
\label{lem:redux-correctness}
Let $(G,d)$ be given and $(G',h)$ be as constructed above. Then $G$ has domatic number at least $d$ if and only if $G'$ has Hadwiger number at least $h$.
\end{lemma}

\begin{proof}
First, suppose that $G$ has domatic number at least $d$; we must show that in this case the Hadwiger number is at least $h$. We can form a family of mutually-adjacent connected subgraphs $S_i$ in $G'$, as follows: for each bottom vertex $b_{j,k}$ form a subgraph consisting of that single vertex, and for each dominating set $D_i$ form a subgraph $S_i$ consisting of a single top vertex together with the middle vertices in $G'$ that correspond to vertices in $D_i$. There are $n(n+1)$ bottom vertices, and $d$ sets containing a top vertex, so these subgraphs form a clique minor with $n(n+1)+d=h$ vertices as desired.

Conversely, suppose that $G'$ has Hadwiger number at least $h=n(n+1)+d$; that is, that it has this many disjoint mutually-adjacent connected subgraphs $S_i$; we must show that, in this case, $G$ has domatic number at least $d$. Each subgraph $S_i$ must include exactly one top or bottom vertex by Lemma~\ref{lem:topbot}, together with possibly some middle vertices. For each top vertex $t_i$, the set $D_i$ is a dominating set in $G$ by Lemma~\ref{lem:dom}; these sets are disjoint because they correspond to the disjoint partition of the middle vertices in $G'$ given by the subgraphs $S_i$.
Thus, we have found $d$ disjoint dominating sets $D_i$ in $G$, so $G$ has domatic number at least $d$.
\end{proof}

\begin{theorem}
The Hadwiger number problem is NP-complete.
\end{theorem}

\begin{proof}
The construction of $(G',h)$ from $(G,d)$ may easily be implemented in polynomial time, and by Lemma~\ref{lem:redux-correctness} it forms a valid polynomial-time many-one reduction from the domatic number problem to the Hadwiger number problem. This reduction (together with the known fact that domatic number is NP-complete and the easy observation that the Hadwiger number problem is in NP) completes the proof of NP-completeness.
\end{proof}

\section{Alternative reduction from disjoint paths}

Seymour~\cite{Seymour} has suggested that it should be straightforward to prove NP-completeness of the Hadwiger number problem via an alternative reduction, from disjoint paths. We briefly outline such a reduction here. In the vertex-disjoint paths problem~\cite{RobSey-JCTB-95}, the input consists of a graph $G$ and a collection of pairs of vertices $(s_i,t_i)$ in $G$; the output should be positive if there exists a collection of vertex-disjoint paths in $G$ having each pair of terminals as endpoints, and negative otherwise. Although the vertex-disjoint paths problem is fixed-parameter tractable with the number of terminal pairs as parameter, it is NP-complete when this number may be arbitrarily large, even when $G$ is cubic and planar~\cite{MidPfe-C-93}.

An instance of this problem may be reduced to the Hadwiger number problem as follows. Let $n$ be the number of vertices in $G$, and $k$ be the number of terminal pairs in the instance. Let $K$ be an $(n+1)$-clique from which $k$ non-adjacent edges $u_iv_i$ have been removed, and form a new graph $G'$ with $2n+1-2k$ vertices as a union of $G$ and $K$ in which $u_i$ is identified with $s_i$ and $v_i$ is identified with $t_i$. Then, a positive solution to the disjoint paths problem in $G$ leads to the existence of an $(n+1)$-vertex clique minor in $G'$, by using the paths to replace each missing edge. Conversely, if $G'$ has an $(n+1)$-vertex clique minor, corresponding to a collection of $n+1$ disjoint connected and pairwise adjacent subgraphs of $G'$, then each of these subgraphs must contain exactly one vertex of $K$ (for any set of vertices of $G'\setminus K$ has at most $n-1$ neighbors), and the adjacency between the two subgraphs containing $u_i$ and $v_i$ can be used to find a path in $G$ connecting $s_i$ and $t_i$ that uses only vertices drawn from these two subgraphs. Therefore, the given vertex-disjoint paths problem instance $(G,k)$ is a positive instance if and only if $(G',n+1)$ is a positive instance of the Hadwiger number problem.

\section{Conclusions}

We have shown that the Hadwiger number problem is NP-complete. It is natural to ask whether the problem is also hard to approximate. Very strong inapproximability results are known for the superficially similar maximum clique problem~\cite{Has-AM-99}.  Alon et al.~\cite{AloLinWah-TCS-07} have provided upper bounds that show that such strong results cannot be true for the Hadwiger number, but they do not rule out the possibility of weaker inapproximability results. In response to an earlier version of the results presented here, Wahlen~\cite{Wah-TCS-09} has provided a preliminary result of this type: unless P${}={}$NP, there can be no polynomial-time approximation scheme for the Hadwiger number. However, there still remains a large gap between this lower bound and the upper bound of $O(\sqrt n)$ on the approximation ratio provided by Alon et~al.

\bibliographystyle{abbrv}
\bibliography{hh}
\end{document}